\relax
\documentclass[11pt, letterpaper]{article}
\usepackage[utf8]{inputenc}
\usepackage[english]{babel}
\usepackage{csquotes}
\usepackage{titling}
\usepackage{amsmath}
\usepackage{hyperref}
\usepackage{a4wide}
\usepackage{cleveref} 
\usepackage{amssymb}
\usepackage{algorithm}
\usepackage{algorithmic}
\usepackage{amsthm}
\usepackage{cite}
\usepackage{todonotes}
\usepackage{tikz}
\usepackage{pgfplots}
\usepackage{thmtools}
\usepackage{authblk}

\newtheorem{theorem}{Theorem}

\newtheorem{definition}{Definition}
\newtheorem{lemma}{Lemma}
\newtheorem{corollary}{Corollary}

\title{A Novel Prediction Setup for Online Speed-Scaling}
\author[1]{Antonios Antoniadis}
\author[2]{Peyman Jabbarzade Ganje}
\author[3]{Golnoosh Shahkarami}
\affil[1]{University of Twente, \href{mailto:a.antoniadis@utwente.nl}{a.antoniadis@utwente.nl}}
\affil[2]{Sharif University of Technology,
\href{mailto:jabarzade@ce.sharif.edu}{jabarzade@ce.sharif.edu}}
\affil[3]{Max Planck Institut für Informatik, Universität des Saarlandes, \href{mailto:gshahkar@mpi-inf.mpg.de}{gshahkar@mpi-inf.mpg.de}}

\date{}
\begin{document}
\maketitle

\begin{abstract}
    Given the rapid rise in energy demand by data centers and computing systems in general, it is fundamental to incorporate energy considerations when designing (scheduling) algorithms. Machine learning can be a useful approach in practice by predicting the future load of the system based on, for example, historical data. However, the effectiveness of such an approach highly depends on the quality of the predictions and can be quite far from optimal when predictions are sub-par. On the other hand, while providing a worst-case guarantee, classical online algorithms can be pessimistic for large classes of inputs arising in practice.
    
    This paper, in the spirit of the new area of machine learning augmented algorithms, attempts to obtain the best of both worlds for the classical, deadline based, online speed-scaling problem: Based on the introduction of a novel prediction setup, we develop algorithms that (i) obtain provably low energy-consumption in the presence of adequate predictions, and (ii) are robust against inadequate predictions, and (iii) are smooth, i.e., their performance gradually degrades as the prediction error increases.
 
\end{abstract}

\section{Introduction}
Energy is a major concern in society in general and computing environments in particular. Indeed, data centers alone are estimated to consume 200 terawatt-hours (TWh) per year, which is likely to increase by a factor of $15$ by year 2030~\cite{nature2018}.
Hardware manufacturers approach this problem by incorporating energy-saving capabilities into their hardware, with the most popular one being \emph{dynamic speed scaling}, i.e., one can adjust the speed of the processor or device. A higher speed implies a higher energy consumption but also more processing capacity. In contrast, a lower speed incurs energy savings while being able to perform less processing per unit of time. Naturally, to take advantage of this energy-saving capability, scheduling algorithms need to decide on what speed to use at each timepoint and consider the energy consumption of the produced schedule alongside more ``traditional'' quality-of-service considerations.
 
This paper studies online, deadline-based \emph{speed-scaling} scheduling, augmented with machine-learned predictions. More specifically, a set of jobs $\mathcal{J}$, each job $j\in\mathcal{J}$ with an associated release time $r_j$, deadline $d_j$ and processing requirement $w_j$, arrives online and has to be scheduled on a single speed-scalable processor. A scheduling algorithm needs to decide for each timepoint $t$ on: (i) the processor speed $s(t)$ and (ii) which job $j\in\mathcal{J}$ to execute at $t$ ($j(t)$). Both decisions have to be made by the algorithm at any timepoint $t$ while only having knowledge of the jobs with a release time equal to or less than $t$. A schedule is said to be feasible if the whole processing requirement of every job $j$ is executed within the respective release time and deadline interval, i.e., if $\int_{t: j(t) = j} s(t) dt \ge w_j$. 
The energy consumption of a schedule, which we seek to minimize over all feasible schedules, is given by $\int_{0}^{+\infty} s(t)^\alpha dt$, where $\alpha>1$ is a constant, which in practice is between $1.1$ and $3$ depending on the employed technology~\cite{CritchlowDS00,WiermanAT12}. 
The offline setting of the problem in which the complete job set $\mathcal{J}$ including their release times, deadlines, and workloads are known in advance was solved in the seminal paper by Yao, Demers, and Shenker~\cite{YDS} who gave an optimal offline algorithm called YDS. The arguably more interesting online setting in which the characteristics of a job $j$ only become known at its release time $r_j$ has been extensively studied~\cite{YDS, Bansal-AVR, BKP,BansalCKP12,AlbersAG15}, and the currently best known online algorithm is qOA, by Bansal et al.~\cite{BansalCKP12} achieving a competitive ratio of $4^\alpha/(2e^{1/2}\alpha^{1/4})$. 

However, the purely online setting may be too restrictive in many practical scenarios for which one can predict -- with reasonable accuracy -- the characteristics of future jobs, for example, by employing a learning approach on historical data.
\emph{Learning augmented algorithms} is a very novel research area (arguably first introduced in 2018 by Lykouris and Vassilvitskii~\cite{lykouris2018competitive}) trying to capture such scenarios in which predictions of uncertain quality are available for future parts of the input. The goal in learning augmented algorithms is to design algorithms that are at the same time (i) \emph{consistent}, i.e., obtain an improved competitive ratio in the presence of adequate predictions,
(ii) \emph{robust}, i.e., there is a worst case guarantee independently of the prediction accuracy (ideally within a constant factor of the competitive ratio of the best known online algorithm that does not employ any predictions) and (iii) \emph{smooth}, i.e., the performance guarantee degrades gracefully with the quality of the predictions.

\paragraph*{Previous Predictions Setups and Our Setup.} 
Online Speed-Scaling with machine learned predictions has been investigated before by Bamas et al.~\cite{BamasMRS20} who consider a prediction setup in a sense orthogonal to ours; the release times and deadlines of jobs are known in advance, and there is a prediction on the processing requirement. Although any input instance (with integer release times and deadlines) can be modeled in such a way (by considering all possible pairs of release times and deadlines and a processing requirement of zero for the pairs that do not correspond to a job), this can be computationally quite expensive. Bamas et al. present a consistent, robust, and smooth algorithm for the particular case in which the interval length of each job is the same.  And generalize their consistency and robustness results to the general case (in which each job can have an arbitrary interval length). For this more general setting, the proof of smoothness is omitted because ``\dots the prediction model and the measure of error quickly get complex and notation heavy''.

In the current paper, we consider the novel prediction setup in which predictions on the release times and deadlines are provided to the algorithm.  To keep the model simple, we assume that the actual processing requirement of each job $j\in\mathcal{J}$, as well as the number of jobs $n$ are known. It may be useful for the reader to think about our setup as having as many unit-size jobs as total processing volume in the instance, and a prediction on the release time and deadline of each such job. We note, however, that our actual setup requires significantly fewer predictions than this simplified one.

In this context, the main contribution of the current paper is to introduce a natural alternative prediction setup and error measure as well as an algorithm (\textsc{SwP}) within that setup, which possesses the desired properties of consistency, smoothness, and robustness in the general setting. It should be pointed out that since the two papers consider different prediction settings and in turn also error measures, the algorithms as well as their guarantees are incomparable. However one can consider the two prediction setups as complementary of each other.

\paragraph*{Our Contribution.} We show how the predictions can be used to develop an algorithm called \textsc{Scheduling with Predictions (SwP)}, that improves upon qOA when the predictions are reasonably accurate. More formally, in Section~\ref{sec:general} we show the following theorem: 
 
\begin{theorem}
    \label{thm:main_general_case}
Algorithm \textsc{SwP} achieves a  competitive ratio of
\begin{align*}
    \min\left\{\left(\frac{1}{1-\mu} \right)^{\alpha-1} \left(\frac{2\eta + 1}{1-2\lambda}\right)^{\alpha-1} , 2^{\alpha-1}\alpha^\alpha \left(\frac{1}{\mu}\right)^{\alpha-1}\right\}.
\end{align*}
\end{theorem}
 
Here, $\eta$ is the \emph{error} of the prediction (defined formally later) that captures the distance between the predicted and the actual input instances, and $0\le\lambda<1/2,0\le\mu\le 1$ are two hyperparameters that can be thought of as the confidence in the prediction. 
Theorem~\ref{thm:main_general_case} implies that \textsc{SwP} is at the same time \emph{consistent}, \emph{smooth} and \emph{robust} where the exact consistency, smoothness and robustness depend on the choice of the hyperparameters $\lambda$ and $\mu$.
 
Additionally, in Section~\ref{sec:common-deadline} we obtain improved results for the restricted case in which all jobs have a common deadline $d$, and we are given predictions regarding the release times of the jobs. The corresponding algorithm is called \textsc{Common-Deadline-Scheduling with Predictions (CDSwP)} and obtains the following improved competitive ratio:
  
  \begin{theorem}
  \label{thm:common-main}
    Algorithm \textsc{CDSwP} achieves a competitive ratio of
    \begin{align*}
        \min\left\{\left( \frac{1+\eta}{1-\lambda}\right)^{\alpha -1},2^\alpha\left(\frac{1+\lambda}{1-\lambda}\right)^{\alpha -1}\right\}.
    \end{align*}
  \end{theorem}
  
  Although restricted, this case seems to capture the difficulty of the online setting for the problem, as supported by the fact that the strongest lower bound of $e^{\alpha -1}/\alpha$ on the competitive ratio for online algorithms for the problem is proven on such an instance~\cite{BansalCKP12}. 
  
  Finally, in Section~\ref{sec:experiments} we present an empirical evaluation of our results on a real-world data-set, which suggests the practicality of our algorithm. The actual results are preceded by Section~\ref{sec:prelim} which contains preliminary results and observations. All omitted proofs can be found in the supplementary material.

\subsection{Related Work.}

\subsubsection{Online Energy-Efficient Scheduling.}

As already mentioned, speed-scaling was first studied from an algorithmic point of view by~\cite{YDS}. They studied the deadline-based version of the problem also considered here, and in addition to providing the optimal offline algorithm called $YDS$, two online algorithms called \emph{Optimal Available (OA)} and \emph{Average Rate (AVR)}. OA recalculates an optimal offline schedule for the remaining instance at each release time, whereas AVR "spreads" the processing volume equally between its release time and deadline in order to determine the speed for each timepoint $t$. The actual schedule then is simply an \emph{Earliest Deadline First (EDF)} schedule with these speeds. They show that AVR obtains a competitive ratio of $2^{\alpha-1}\alpha^\alpha$ which is essentially tight as shown by~\cite{Bansal-AVR}. Algorithm OA, on the other hand, was analyzed by~\cite{BansalCKP12} who proved a tight competitive ratio of $\alpha^\alpha$.

The currently best known algorithm for the problem, at least for modern processors which satisfy $\alpha=3$ is the aforementioned $qOA$ algorithm, which for any parameter $q\ge 1$ sets the speed of the processor to be $q$ times the speed that the optimal offline algorithm would run the jobs in the current state. Algorithm qOA attains a competitive ratio of $4^\alpha/(2e^{1/2}\alpha^{1/4})$,  for $q= 2-1/\alpha \approx 1.667$.

The multiprocessor version of online, deadline-based, speed-scaling has also been studied, see~\cite{AlbersAG15, AngelBKL19} as well as other objectives, for example, flow time~\cite{AlbersF07,BansalPS09}. We refer the interested reader to surveys~\cite{Albers10,IraniP05}.

\subsubsection{Further Results on Learning Augmented  Algorithms.}
\cite{lykouris2018competitive} was arguably the seminal paper in the area, considered the online caching problem. Subsequently, \cite{NEURIPS2018_73a427ba} considered the ski-rental problem as well as non-clairvoyant scheduling. Similar to the current work, the robustness and consistency guarantees were given as a function of a hyperparameter that is part of the input to the algorithm. Both the caching and the ski-rental problem have since been extensively studied in the literature (see for example~\cite{rohatgi2020near, antoniadis2020online, Wei20} and~\cite{WangL20, pmlr-v97-gollapudi19a}).

Several other online problems have been investigated through the lens of learning-augmented algorithms and results of similar flavor were obtained. Examples include scheduling and queuing problems~\cite{48659, Mitzenmacher20}, online selection and matching problems~\cite{AntoniadisGKK20, DuttingLLV21}, or the more general framework of online primal-dual algorithms~\cite{BamasMS20}. We direct the interested reader to a recent survey~\cite{MitzenmacherV20}.

\section{Preliminaries}
\label{sec:prelim}
We consider \emph{online, deadline-based speed-scaling} as described in the introduction. 
Given a scheduling algorithm $A$ on the set jobs $\mathcal{J}$, the energy consumption of $A$ on $\mathcal{J}$ is denoted by $\mathcal{E}_A (\mathcal{J})$.
When clear from the context, we may write $\mathcal{E}_A$ instead of $\mathcal{E}_A (J)$ to simplify the notation.

As usual for online problems, the performance guarantees are given by employing \emph{competitive analysis}. Following the speed-scaling literature (see for example~\cite{BansalCKP12}) we use the $\emph{strict competitive ratio}$. Formally, the (strict) competitive ratio of algorithm $A$ for the online, deadline-based speed-scaling problem, on input instance $I$ is given by
\begin{align*}
    \max_I \frac{\mathcal{E}_{A(I)}}{\mathcal{E}_{YDS(I)}},
\end{align*}

were $\mathcal{E}_{A(I)}$ is the cost that algorithm $A$ incurs on instance $I$, and the maximum is taken over all possible input instances $I$. 
The competitive ratio in many cases will depend on the prediction error.

\paragraph{Prediction Setup.} 
The algorithm initially gets information about the number of jobs $n$, the corresponding processing volumes $w_j, \forall j\in \mathcal{J}$, as well as for every job $j\in\mathcal{J}$ a prediction $p_j$ for the release time $r_j$ and another prediction $q_j$ for the deadline $d_j$. Again, the actual values of $r_j$ and $d_j$ only become known at timepoint $r_j$. Let $R=\{r_1,\ldots r_n\}$, $D=\{d_1,\ldots d_n\}$, $P=\{p_1,\ldots p_n\}$ and $Q=\{q_1,\ldots q_n\}$. Note that in the special case where all jobs have a common deadline $d$ we naturally only obtain predictions for the release times.

The quality of the prediction is measured in terms of a \emph{prediction error} $\eta$, which intuitively $\eta$ measures the distance between the predicted values and the actual ones. We start by defining the individual prediction error $\eta_i$ for each job $i\in\mathcal{J}$.

\begin{definition}
Let the prediction error for job $i$ be $\eta_i = \max \left\{ \frac{|p_{i} - r_i|}{q_{i} - p_i}, \frac{|q_{i} - d_i|}{q_{i} - p_i}\right\}$.
\end{definition}

Note that we implicitly assume that $p_i\leq q_i$ for all $i\in\mathcal{J}$ since otherwise, it is immediately obvious that the quality of the predictions is low, and one could just run a classical online algorithm for the problem. Furthermore, if the instance has a common deadline then $\eta_i$ simplifies to $\eta_i =  \frac{|p_{i} - r_i|}{d - p_i}$.

The (total) prediction error $\eta$ of an input instance is then given by $\eta = \max_i\eta_i$.  We call this \emph{max-norm-error}.

\begin{definition}
    We say that the total error $\eta$ is a \emph{max-norm-error} if $\eta$ is given by the infinity norm of the vector of the respective errors for each job. More formally, 
    \begin{align*}
        \eta = \lVert\boldsymbol{\eta} \rVert_\infty = \max (\eta_1,\eta_2,\dots \eta_n).
    \end{align*}
\end{definition}

\paragraph{Performance Guarantees.} 
In the following we formalize the performance guarantees used to evaluate our algorithms.
\begin{definition}
    We say that an algorithm within the above prediction setup is:
    \begin{itemize}
        \item \emph{Consistent,} if its competitive ratio is strictly better than that of the best online algorithm without predictions for the problem, whenever $\eta = 0$.
        \item \emph{Robust,} if its competitive ratio is within a constant factor from that of the best online algorithm without predictions for the problem. Note that Robustness is independent of the prediction quality.
        \item \emph{Smooth,} if its competitive ratio is a smooth function of $\eta$.
    \end{itemize}
\end{definition}

\paragraph{Shrinking of Intervals.}
The most straightforward way to consider the predictions would arguably be to blindly trust the predictions, i.e., schedule jobs assuming that the predicted instance is the actual instance. 

Consider the instance $J_{PQ}$ (resp. $J_{RD}$) in which every job has the corresponding predicted (resp. actual) release time and deadline. The naive algorithm would compute the optimal offline schedule $YDS(J_{PQ})$ and try to schedule tasks according to it. If the predictions are perfectly accurate, then this clearly is an optimal schedule, and the best one can do. However, if the predictions are even slightly inaccurate, then the resulting schedule may be infeasible. 
Moreover, our goal is to have a robust algorithm, which cannot be obtained by following the predictions blindly.
For these reasons, one has to trust the predictions more cautiously and not blindly.

One of our crucial ideas is to slightly shrink the interval between each job's release time and deadline before scheduling it. The intuition is that if the predictions are only slightly off, then a YDS schedule for the newly obtained instance will be feasible at a slight increase in energy consumption over the YDS schedule of the predicted instance. The following lemmas formalize this intuition. We note that a similar result is also presented in~\cite{BamasMRS20}; however, given that the actual setups are different new proofs are required (the proofs of these lemmas can be found in the supplementary material).

\begin{restatable}{lemma}{ShrinkfromOneSide}
\label{lem:ShrinkfromOneSide}
    Consider a common deadline instance $\mathcal{J}$, and another common deadline instance $\mathcal{\hat{J}}$ constructed from $\mathcal{J}$ such that every job $\hat{j}_i \in\mathcal{\hat{J}}$ has workload $\hat{w}_i=w_i$, $\hat{d}=d$, and $\hat{r}_i= r_i + (1-c_i) \cdot (d - r_i)$ for some shrinking parameter $0\le c_i<1$. Set $c=\max_i c_i$. Then, 
    \begin{align*}
        \mathcal{E}_{YDS}(\mathcal{\hat{J}})  \le (1/c)^{\alpha-1}\mathcal{E}_{YDS}({\mathcal{J}}).
    \end{align*}
\end{restatable}

\begin{restatable}{lemma}{ShrinkingRatio}
\label{lem:ShrinkingRatio}
Consider a (general) instance $\mathcal{J}$, and another  instance $\mathcal{\hat{J}}$ in which every job $j_i\in \mathcal{J}$ corresponds to a job $\hat{j}_i\in \mathcal{\hat{J}}$ with workload $\hat{w}_i=w_i$,  $\hat{r}_i= r_i + \frac{1-c}{2}\cdot (d_i-r_i)$ and $\hat{d}_i = d_i -\frac{1-c}{2}\cdot (d_i-r_i)$ for some shrinking parameter $0\le c<1$. Then, 
    \begin{align*}
        \mathcal{E}_{YDS}(\mathcal{\hat{J}}) \le (1/c)^{\alpha-1}\mathcal{E}_{YDS}(\mathcal{J}).
    \end{align*}
\end{restatable}

It will be useful to bound the energy consumption of (the possibly infeasible for the original input instance) schedule $YDS(J_{PQ})$.
We compute the energy consumption of schedule $YDS(J_{PQ})$ in the following lemma.

\begin{restatable}{lemma}{PredError}
\label{lem:PredictionError}
    For any $\eta\ge 0$ there holds 
        \begin{align*}
        \mathcal{E}_{YDS(J_{PQ})} \le (2\eta + 1)^{\alpha -1} \mathcal{E}_{YDS(J_{RD})}.
    \end{align*}
\end{restatable}
\begin{proof}
    Consider two sets $P^* = \{p^*_1,\dots p^*_n\}$ and $Q^* = \{q^*_1,\dots q^*_n\}$, with 
    $p^*_i = p_i - \eta_i (q_i - p_i)$ and $q^*_i = q_i + \eta_i (q_i - p_i)$.
    
    By the definition of $\eta_i$, $p_i^*$ and $q_i^*$, we have $(r_i,d_i) \subseteq (p^*_i,q^*_i)$, and therefore

    \begin{align*}
        \mathcal{E}_{YDS(J_{P^*Q^*})} \le \mathcal{E}_{YDS(J_{RD})}.
    \end{align*}

    By having $c = \frac{1}{(2\eta + 1)}$, and $J= J_{P^*Q^*}$ $(r_i = p^*_i, d_i = q^*_i)$ in Lemma~\ref{lem:ShrinkingRatio}, we obtain $J' = J_{PQ}$ and therefore,
    
    \begin{align*}
      \mathcal{E}_{YDS(J_{PQ})} \le (2\eta + 1)^{\alpha -1} \mathcal{E}_{YDS(J_{P^*Q^*})}.
    \end{align*}
\end{proof}

Using Lemma~\ref{lem:ShrinkfromOneSide}, we can obtain a similar result for common deadline instances.

\begin{corollary}\label{CD_PredictionError}
In common deadline instances for any parameter $\eta\ge 0$, there holds 
\begin{align*}
      \mathcal{E}_{YDS(J_{P})}\leq  (\eta +1)^{\alpha-1} \cdot  \mathcal{E}_{YDS(J_{R})}.
\end{align*}
\end{corollary}

The idea of shrinking intervals as described above will be useful for the general case as well as the restricted common deadline case.

How much each algorithm will shrink the predicted job intervals will depend on the \emph{confidence}. This will be denoted by a \emph{confidence parameter} $0< \lambda \le 1/2$ that will be given as input to the respective algorithm. In the following, we define the \emph{shrunk prediction set} of release times and deadlines parametrized by this $\lambda$, and use the above lemmas to argue about how this "shrinking" actually affects the energy consumption of the corresponding $YDS$-schedule.

\begin{definition}\label{shrunk Predictions}
Let $P'=\{p'_1, \ldots, p'_n\}$ and $Q'=\{q'_1, \ldots, q'_n\}$ be the shrunk prediction set of release times and deadlines respectively in which $p_i' = \lfloor p_i + \lambda (q_i-p_i)\rfloor$ and $q_i' = \lceil q_i -\lambda (q_i-p_i) \rceil$ for all $i \in [n]$.
\end{definition}

We first observe that any schedule that considers the sets $P'$ and $Q'$ as the actual release times and deadlines of the jobs will be feasible, as long as the error $\eta$ is not larger than $\lambda$.

\begin{restatable}{observation}{pjrjqjdj}
    \label{obs:pjrjqjdj}
    Under the assumption that $\eta\in(0,\lambda)$, it follows that $r_i \le p_i'$ and $q_i' \le d_i$ hold for every job $i$.
\end{restatable}
\begin{proof}
    If $p_j \ge r_j$ the observation directly follows because $p_j'\ge p_j$.
    
    If $p_j< r_j$, let $p_i'' =  p_i + \lambda (q_i-p_i)$.  By the assumption and the definitions of $\eta$, and $\eta_i$, we have
    \begin{align*}
        |p_i-r_i| \le \lambda (q_i -p_i) = p_i''-p_i.
    \end{align*}
    
     Then the above equation gives
        \begin{align*}
            -p_i + r_i \le p_i'' -p_i, 
        \end{align*}
        and therefore $ r_i \le p_i''$. Since $p_i' = \lfloor p_i'' \rfloor$ and $r_i$ is an integer, we can also conclude $ r_i \le p_i'$. 
        
        The same holds for the deadlines and their shrunk predictions.
\end{proof}

Therefore, the schedule $YDS(J_{P'Q'})$ is feasible. Although shrinking the intervals and then running YDS is not a robust algorithm, it will be useful to bound its energy consumption when $\eta\le\lambda$ holds.

\begin{restatable}{lemma}{shrunkpredictionerror}
\label{shrunkPredictionError}
For any $\eta\in(0,\lambda)$ there holds 
\begin{align*}
        \mathcal{E}_{YDS(J_{P'Q'})}\leq \left(\frac{2\eta + 1}{1-2\lambda}\right)^{\alpha-1}  \cdot  \mathcal{E}_{YDS(J_{RD})}.
    \end{align*}
\end{restatable}
\begin{proof}
    \begin{align}
    \label{eq:m2}
        \mathcal{E}_{YDS(J_{P'Q'})}\leq \frac{1}{(1-2\lambda)^{(\alpha-1)}}  \cdot  \mathcal{E}_{YDS(J_{PQ})} \nonumber \\
        \leq \left(\frac{2\eta + 1}{1-2\lambda}\right)^{\alpha-1}  \cdot  \mathcal{E}_{YDS(J_{RD})}.
    \end{align}
    By Lemma~\ref{lem:ShrinkingRatio} we have the first inequality in (\ref{eq:m2}), and the second inequality holds because of Lemma~\ref{lem:PredictionError}.
\end{proof}

Similarly for common deadline instances, since we shrink from one side, we obtain a better competitive ratio. 

\begin{restatable}{corollary}{CDshrunk}
\label{CD_shrunkPredictionError}
For any $\eta\in(0,\lambda)$ in common deadline instances there holds 
\begin{align*}
        \mathcal{E}_{YDS(J_{P'})}\leq \left(\frac{1+ \eta}{1-\lambda}\right)^{\alpha-1}  \cdot  \mathcal{E}_{YDS(J_{R})}.
    \end{align*}
\end{restatable}
\begin{proof}
    \begin{align}
    \label{eq:CD_m2}
        \mathcal{E}_{YDS(J_{P'})}\leq \frac{1}{(1-\lambda)^{(\alpha-1)}}  \cdot  \mathcal{E}_{YDS(J_{P})} \nonumber \\
        \leq \left(\frac{1 + \eta}{1-\lambda}\right)^{\alpha-1}  \cdot  \mathcal{E}_{YDS(J_{R})}.
    \end{align}
    By Lemma~\ref{lem:ShrinkfromOneSide} we have the first inequality in (\ref{eq:CD_m2}), and the second inequality holds because of Corollary \ref{CD_PredictionError}.
\end{proof}

\section{General Case}
\label{sec:general}
In this section we present algorithm \textsc{ScheduleWithPredictions}$(\lambda,\mu)$ (\textsc{SwP}$(\lambda,\mu)$ for short) for the general learning-augmented speed-scaling setting.
Parameter $0 \le \lambda < 1/2$  describes for which range of prediction errors we would like to obtain an improved competitive ratio. The smaller the $\lambda$, the smaller that range but the better the corresponding competitive ratio for $\eta<\lambda$. On the other hand, parameter $0\le \mu\le 1$ allows us to set the desired trade-off between consistency and robustness. As we will see, perfect predictions and $\lambda = \mu = 0$ would give a competitive ratio of $1$.

Inspired by~\cite{BamasMRS20} algorithm \textsc{SwP} 
begins by partitioning each time slot $I_t = [t,t+1), t\in\mathbb{Z}$  into two parts: $I_t^\ell = [t, t+(1-\mu))$ and $I_t^r = [t+(1-\mu),t+1)$. We call $I_t^\ell$ the \emph{left part}, and $I_t^r$ the \emph{right part} of time slot $I_t$. The idea is to reserve the left parts of time-slots for following the prediction, and the right parts of the time-slots are, roughly speaking, intended for safeguarding against inaccurate predictions. A key component of our algorithm consists of elegantly and dynamically distributing the processing volume of each job upon its arrival among the two parts. This distribution is crucial in order to obtain a trade-off between consistency and robustness, based on the parameters $\lambda$ and $\mu$. The algorithm consists of two steps, the \emph{preprocessing} and the \emph{online} step which we now describe in more detail.

\paragraph{Preprocessing: Partition left parts into intervals and assign jobs to them.}
Upon receiving the predictions $(P,Q)$, \textsc{SwP} computes a YDS-schedule $S'$ for instance $(P',Q')$ -- which is obtained by "shrinking" $(P,Q)$ as described above. Although $S'$ may not be feasible for the actual instance $(R,D)$, it will be used to partition the left parts into intervals and subsequently assign each such interval of the partition to a specific job. 

To this end, let $I_t(j) := [t + a_t(j),t + b_t(j))\subseteq I_t$ be the maximal subinterval of $I_t$ during which $j$ is executed under $S'$. Note that $I_t(j)$ could be empty for some combinations of $j$ and $t$. Furthermore, since by definition there are no release times or deadlines within $I_t$, and YDS schedules according to EDF, there can be at most one execution interval of $j$ within $I_t$. Let, for every job $j$ and left part $I_t^\ell$, $I_t^\ell(j):= [t+a_t^\ell(j),t+b_t^\ell(j))$, where $a_t^\ell(j) = a_t(j)/(1-\mu)$ and $b_t^\ell(j)=b_t(j)/(1-\mu)$, be the subinterval of $I_t^\ell$ assigned to job $j$.

To obtain some intuition, scheduling the whole processing volume of each job $j$ at a uniform speed throughout intervals $I_t^\ell(j)$  would result in a "compressed" version of $YDS(P'Q')$ where each time-slot is sped-up by a factor of $1/(1-\mu)$ to fit in the left part only, thus having an energy consumption increased by a factor of $(1/(1-\mu))^\alpha$ over that of $YDS(P'Q')$. Although (as we will see) such a compressed schedule would be consistent, it may not be robust (or even feasible) in the presence of subpar predictions. For this reason, we will eventually only schedule part of the volume of each job in the associated left parts whenever feasible, and the remaining volume will be processed on right parts.

\paragraph{Online Step: Job arrivals and processing}
\textsc{SwP} needs to decide exactly when each job is to be processed within each time-slot and at what speed. This is done by (i) distributing the processing volume of each job $j$ to right parts of different time-slots $I_t$ and associated left parts $I_t^\ell(j)$ upon its arrival, and (ii) feasibly scheduling the whole volume assigned to the current time-slot (both to its left and right part), within the time-slot itself. In the following we discuss how this is accomplished.

\subparagraph{(i) Job Arrivals:} Upon arrival of job $j$ at $r_j$, let $\delta_j = w_j/(d_j-r_j)$ be its density and $ \ell(j) := \sum_{t\in[r_j,d_j)} |I_t^\ell(j)|$  be the total processing time reserved for job $j$ on the left parts during the preprocessing step that can actually be feasibly used for job $j$. Furthermore let $V_t(j)$ be the total volume currently (from jobs $1,2,\dots j-1$) assigned to $I_t^r$, for all $t$ (thus $V_t(1) = 0)$. 

The algorithm assigns some amount of volume $y_j^t$ (to be determined later) of job $j$ to interval $I_t^r$ (thus $V_t(j+1) := V_t(j) + y_j^t$), for all $t\in [r_j,d_j)$, with $0 \le y_j^t \le \delta_j$.
Finally the remaining volume $X_j := w_j - \sum_t y_j^t$ is assigned to the left parts $I_t^\ell(j)$ with $t\in[r_j,d_j)$, proportionally to their length, i.e., an interval $I_t^\ell(j)$ with $t\in[r_j,d_j)$ receives an $|I_t^\ell(j)|/\ell(j)$-fraction of $X_j$ which implies that the average speed within $I_t^\ell(j)$ must be $X_j/\ell(j)$. To gain some intuition on the values of $y_j^t$, it is useful to think of the algorithm as waterfilling the volume of $j$ to both the left and the right parts such that no right part receives more than $\delta_j$ amount of volume. More formally, the $y_j^t$, with $0\le y_j^t\le \delta_j$ and $t\in[r_j,d_j)$ are defined such that they satisfy the following inequalities:

 \begin{align}
        & \frac{V_t(j)}{\mu} \ge X_j/\ell(j) & \forall t\in[r_j,d_j) \text{ with } y_j^t = 0 \label{eq:y_full}\\
        & \frac{V_t(j) + y_j^t}{\mu} = X_j/\ell(j)  & \forall t\in[r_j,d_j) \text{ with } 0 < y_j^t < \delta_j \label{eq:y_between}\\
        & \frac{V_t(j) + y_j^t}{\mu} \le X_j/\ell(j) & \forall t\in[r_j,d_j) \text{ with } y_j^t = \delta_j. \label{eq:y_zero}
\end{align} 

Note that the left hand side in each of the above inequalities corresponds exactly to $V_t(j+1)/\mu$ and therefore to the average speed required to process the volume assigned to $t$ before the arrival of job $j+1$ within $I_t^r$.
We prove the existence of such $y_j^t$ and describe how they can be computed in Appendix~\ref{app:y}.

\subparagraph{(ii) Processing:} For each $I_t, t=r_j,\dots r_{j+1}-1$ the algorithm processes job $j'\le j$ within every $I_t^\ell(j')$ at a speed of $X_{j'}/\ell(j')$, and the assigned volume to $I_t^r$ is processed within $I_t^r$ at a speed of $V_t(j+1)/\mu$, with the order of the jobs within each $I_t^r$ being determined by EDF.

\begin{figure}[t]

\centering 
\begin{tikzpicture}
\draw[thick,->] (0,0) -- (10,0) node[anchor=north west] {Time};
\draw[thick,->] (0,0) -- (0,4) node[anchor=south east] {Speed};
\foreach \x in {0,1,2,3,4,5,6,7,8,9}
   \draw (\x cm,1pt) -- (\x cm,-1pt) node[anchor=north] {$\x$};
   
\filldraw[fill=red!40!white, draw=black] (2,0) rectangle (2.75,2.5);
\filldraw[fill=gray!40!white, draw=black] (2.75,0) rectangle (3,3.3);
   
\filldraw[fill=red!40!white, draw=black] (3,0) rectangle (3.75,2.5);
\filldraw[fill=gray!40!white, draw=black] (3.75,0) rectangle (4,2.5);

\filldraw[fill=gray!40!white, draw=black] (4.75,0) rectangle (5,2.5);

\filldraw[fill=gray!40!white, draw=black] (5.75,0) rectangle (6,2.5);

\filldraw[fill=red!40!white, draw=black] (6,0) rectangle (6.75,2.5);
\filldraw[fill=gray!40!white, draw=black] (6.75,0) rectangle (7,1.7);
\filldraw[fill=red!40!white, draw=black] (6.75,1.7) rectangle (7,2.5);

\filldraw[fill=gray!40!white, draw=black] (7.75,0) rectangle (8,1.3);
\filldraw[fill=red!40!white, draw=black] (7.75,1.3) rectangle (8,2.5);

\filldraw[fill=gray!40!white, draw=black] (8.75,0) rectangle (9,0.5);
\filldraw[fill=red!40!white, draw=black] (8.75,0.5) rectangle (9,2.5);

\end{tikzpicture}

\caption{The speed profile corresponds to an instance with $\mu = 0.25$. Job $i$ arrives at $r_i=2$, with $d_i=9$, and $w_i=7$. Hence, $\delta_i = \frac{w_i}{d_i - r_i}=1$. For this instance, $YDS(P'Q')$ runs job $i$ only in 3 blocks, so we have $\ell(i) = 3 \cdot 0.75 = 2.25$. 
For the first four blocks we have $y^t_i = 0$ and inequality~(\ref{eq:y_full}) holds. In the fifth and sixth blocks, $0 < y^t_i < \delta_i$ and inequity~(\ref{eq:y_between}) holds. And in the last block, $y^t_i = \delta_i =1$ and inequality~(\ref{eq:y_zero}) holds.}
\label{fig1}
\end{figure}
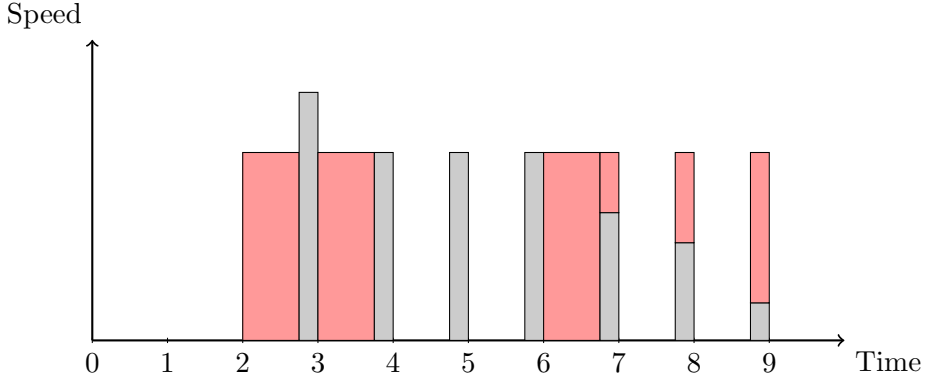

The online step gets repeated upon the arrival of each job. We next show that the resulting schedule is feasible.
\begin{restatable}{lemma}{generalfeasibility}
\label{lem:generalfeasibility}
In the schedule output by $\textsc{SwP}(\lambda,\mu)$ a volume of $w_j$ is fully processed for each job $j$ within $[r_j,d_j)$.
\end{restatable}

\begin{proof}
    It is relatively easy to see that a total volume of $w_j$ is assigned to left and right parts of $I_t$'s with $t\in[r_j,d_j)$: Indeed, by the algorithm definition volume of $w_j$ only gets assigned to  $I_t^\ell\subset I_t$  or $I_t\subset I_t$ with $t\in[r_j,d_j)$. In addition, a volume of $\sum_t y_j^t$ gets assigned to the right parts and $w_j-\sum_t y_j^t$ to the left parts for a total volume assignment of $w_j$. It therefore remains to show that all the assigned volume is feasibly processed in the processing step.
    
    Consider some $I_t$ with $r_j \le t < r_{j+1}$ and the corresponding $I_t^\ell$ and $I_t^r$. Note that by the above argument, for any such $t$ no job with an index greater than $j$ will assign any volume, and that only job $j'\le j$ may be assigned to $I_t^\ell(j')$. Therefore a speed of $X_{j'}/\ell(j')$ throughout every such $I_t^\ell(j)$ is sufficient to schedule all volume assigned to it. Finally, since there are no release times or deadlines within each individual interval, the total volume of $V_t(j+1)$ can be feasibly scheduled within $I_t^r$ at a speed of $V_t(j+1)/\mu$.
\end{proof}

We next show consistency and robustness of the algorithm.

\begin{lemma}[Consistency] 
\label{lem:general-consistency}
For any $\eta\in (0,\lambda)$ there holds
\begin{align*}
     \mathcal{E}_\textsc{SwP}   \le \left(\frac{1}{1-\mu} \right)^{\alpha-1} \left(\frac{2\eta + 1}{1-2\lambda}\right)^{\alpha-1} \mathcal{E}_{YDS(RD)}.
\end{align*}
\end{lemma}
\begin{proof}
        We can express $\mathcal{E}_\textsc{SwP}$  as:
       \begin{align*}
        \mathcal{E}_\textsc{SwP}  = \sum_{j=1}^n \left( \frac{X_j^\alpha}{\ell(j)^{\alpha-1}} + \sum_{t\in[r_j,d_j)}\left(\frac{V_t(j+1)^\alpha}{\mu^{\alpha-1}}- \frac{V_t(j)^\alpha}{\mu^{\alpha-1}}\right)\right)\\
    =         \sum_{j=1}^n \left( \frac{X_j^\alpha}{\ell(j)^{\alpha-1}} + \sum_{t\in[r_j,d_j)}\left(\frac{\left(V_t(j) + y_j^t\right)^{\alpha}}{\mu^{\alpha-1}}- \frac{V_t(j)^{\alpha}}{\mu^{\alpha-1}}\right)\right)\\
\le  \sum_{j=1}^n \frac{w_j^\alpha}{\ell(j)^{\alpha-1}}
        = \left( \frac{1}{1-\mu}\right)^{\alpha-1}\mathcal{E}_{YDS(J_{P'Q'})}^j.
    \end{align*}
 The inequality holds by convexity of the power function and by the fact that $V_t(j+1)/\mu \le X_j/\ell(j)$ for each $t$ such that $y_j^t>0$  (Equations~\ref{eq:y_between} and~\ref{eq:y_zero}). The last equality follows since for $\eta\in (0,\lambda)$, for every job $j$ there holds $[r_j,d_j)\supseteq [p'_j, q'_j)$ (Observation~\ref{obs:pjrjqjdj}), and because by construction $\ell(j)$ is $1/(1-\mu)$ times the total processing time reserved for job $j$ under YDS$(P'Q')$.
    
    The lemma directly follows, since by Lemma~\ref{shrunkPredictionError}, 
    \begin{align*}
        \mathcal{E}_{YDS(J_{P'Q'})}\leq \left(\frac{2\eta + 1}{1-2\lambda}\right)^{\alpha-1}  \cdot  \mathcal{E}_{YDS(J_{RD})}.
        \end{align*}    
  
\end{proof}

\begin{lemma}[Robustness]\label{Robustness}
\label{lem:general-robustness}
    For any instance, we have
    \begin{align*}
        \mathcal{E}_\textsc{SwP} \le 2^{\alpha-1}\alpha^\alpha \left(\frac{1}{\mu}\right)^{\alpha-1}\mathcal{E}_{YDS(J_{RD})}. 
    \end{align*}
\end{lemma}
\begin{proof}
    Note that by the algorithm definition there holds that $V_t(j) \ge V_t(i)$, for $j>i$ and any $t$, since upon each release time new volume gets assigned but volume never gets removed. 
    We therefore have 
    \begin{align*}
        \mathcal{E}_\textsc{SwP} \le 
        \sum_{j=1}^n \left( \frac{X_j^\alpha}{\ell(j)^{\alpha-1}}\right) 
        + \sum_{t}\left(\frac{V_t(n+1)^\alpha}{\mu^{\alpha-1}}\right)\\
        \le  \sum_{t}\left(\frac{\left(\sum_{j:t\in[r_j,d_j)}\delta_j\right)^\alpha}{\mu^{\alpha-1}}\right)\\
        = \left(\frac{1}{\mu}\right)^{\alpha-1}\mathcal{E}_{AVR}.
    \end{align*}
    The second inequality follows by the convexity of the power function and the fact that $V_t(n+1)/\mu\ge V_t(j+1)/\mu\ge X_j/\ell(j)$ for each $t$ such that $y_j^t<\delta_j$ (Equations~\ref{eq:y_between} and~\ref{eq:y_full}).
    The lemma follows by the competitive ratio of AVR~\cite{Bansal-AVR}. 
\end{proof}

Lemmas~\ref{lem:general-consistency} and~\ref{lem:general-robustness} together directly imply Theorem~\ref{thm:main_general_case}. 
Note that Theorem~\ref{thm:main_general_case} not only implies consistency and robustness, but also smoothness: the competitive ratio gracefully degrades as the error increases.

\section{All Jobs Have a Common Deadline}
\label{sec:common-deadline}
In this section, we present a simpler algorithm that achieves improved consistency and robustness over \textsc{SwP} for the special case in which all jobs have the same deadline, i.e., $d_j=d$ for all $j \in \mathcal{J}$.
Since the deadline is the same for all jobs, we only consider predictions on the $n$ release times $R= \{r_1, \ldots, r_n\}$ and denote these by a set $P=\{p_1, \ldots, p_n \}$.

We begin by analyzing a framework for combining different algorithms before presenting an algorithm in Subsection~\ref{sec:cdswp} that is based on combining two different algorithms; the classic online algorithm $qOA$ that has a worst-case guarantee independent of the prediction error, and a second one, that considers the predictions and has a good performance in the case of small prediction error. 

The general idea of combining online algorithms has been repeatedly employed in the past in the areas of online algorithms and online learning, see, for example, the celebrated results of Fiat et al.~\cite{FIAT1994410}, Blum and Burch~\cite{blum2000line}, Herbster and Warmuth~\cite{herbster1998tracking}, Littlestone and Warmuth~\cite{LITTLESTONE1994212}.  Such a technique has also been used in the learning augmented setting, see Antoniadis et al.~\cite{antoniadis2020online} for an explicit framework for combining algorithms, and Lykouris and Vassilvtiskii~\cite{lykouris2018competitive} as well as Rohatgi~\cite{rohatgi2020near} for implicit uses of such algorithm combinations. However, as we will see, the specific problem considered in this paper allows for way more flexibility in such algorithm combinations since it is possible to simulate the parallel execution of different algorithms by increasing the speed. This allows us to obtain a much more tailored result with at most one switch between the different algorithms and more straightforward analysis.
We start with the following structural lemma.

\begin{restatable}{lemma}{merge}
\label{lem:merge}
Consider a partition of the job set of instance $J$ into $m$ job sets $J_1,J_2,\dots J_m$,
 and furthermore consider $m$ schedules $C_1,C_2,\dots C_m$ with speed functions $s_1(t),s_2(t),\dots s_m(t)$ respectively, such that $C_i$ is a feasible schedule for $J_i$ for all $i= 1,\dots m$. Then there exists a schedule $C$  with speed function $s_{C}(t)=\sum_i s_i(t)$ that is feasible for the complete job  set $J$ and has an energy consumption of $\mathcal{E}_{C}\le m^{\alpha-1} \sum_i \mathcal{E}_i$, where for each $i$, $\mathcal{E}_{i} = \int_t s_{i}(t)^\alpha dt $ is the energy consumption of the respective schedule.
\end{restatable}

\subsection{Algorithm \textsc{CommonDeadlineScheduleWithPredictions} (\textsc{CDSwP})}
\label{sec:cdswp}
At a high level \textsc{CDSwP}$(\lambda)$ (almost) follows the optimal schedule for the predicted instance as long as the prediction error is not higher than $\lambda$ and switches to a classical online algorithm (i.e., one without predictions) in case the prediction error becomes higher than $\lambda$. 

More formally,  the algorithm can reside in one of two modes: \emph{follow the prediction} (FtP) mode, and \emph{recovery} mode. Initially, before the release time $r_1$ of the first job the algorithm is in the FtP-mode and has an associated speed-profile given by $s(FtP(0),t) = 0$ for all $t\in[0,d]$. Upon each release time $r_i$, $i=1,\dots n$, and while in the FtP-mode, \textsc{CDSwP}$(\lambda)$ does the following:

\begin{itemize}
   \item  If $\eta_i\le \lambda$, \textsc{CDSwP} remains in the FtP-mode and updates the speed profile from $s(FtP(i-1),t)$ to $s(FtP(i),t)$ for $[r_i,d]$ with the help of a job instance $J^i$. Instance $J^i$ consists of:
    \begin{itemize}
    \item One job $i'$ with release time $r_{i'} = r_{i}$, workload $w_{i'}$ equal to the total amount of unfinished at $r_i$ workload that was released at any timepoint $t\le r_i$, and deadline $d$.
    
    \item For each job $j$ not yet released at $r_j$, include job $j$ with a release time of $p_j'$, a deadline of $d$ and a volume of $w_i$ in $J^i$.
\end{itemize}
The new speed-profile $s(FtP(i),t)$ is given for any $t\in[r_i,d]$ by
\begin{align*}
    s(FtP(i),t) := \begin{cases}
    s(YDS(J'),t), & \parbox[t]{0.34\columnwidth}{\text{if} $YDS(J')$ \text{runs job} $i'$ \text{at} $t$,}\\
    0, &\text{otherwise}.
    \end{cases}
\end{align*}

Algorithm \textsc{CDSwP} now runs at $s(FtP(i),t)$ for any $t\in [r_i, r_{i+1})$, and remains in the FtP-mode. 

\item     Otherwise, if $\eta_i>\lambda$ then \textsc{CDSwP} switches to the recovery-mode, and sets $k:= i$. 
   \end{itemize}
   
   When in recovery-mode, the algorithm runs at  speed $s(t) = s(FtP(k-1),t) + s(qOA(k),t)$ at each timepoint $t$ until $d$, where $s(FtP(k-1),t)$ is the last speed-profile generated in the FtP-mode, and  $s(qOA(k),t)$, is the speed that the online algorithm $qOA$ would have at timepoint $t$ when presented (in an online fashion) with (the actual) jobs $k,\dots n$. 
   
    Note that defining the speed at any timepoint $t$ is sufficient in order to fully describe the algorithm. Indeed, since all jobs have a common deadline of $d$, it is irrelevant which job (among the active jobs) is being processed at any timepoint $t$. Nevertheless, to simplify the presentation we will implicitly assume in the following that at timepoint $t$ the currently active and unfinished job with the earliest release time is the one being processed -- and ties are broken arbitrarily. We first prove that the algorithm produces feasible schedules:
    
    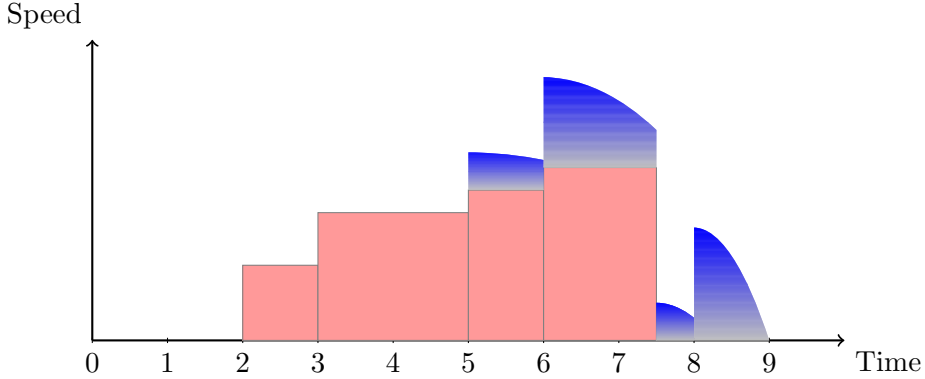
\begin{figure}[t]
        \centering
        \begin{tikzpicture}
        \draw[thick,->] (0,0) -- (10,0) node[anchor=north west] {Time};
        \draw[thick,->] (0,0) -- (0,4) node[anchor=south east] {Speed};
        \foreach \x in {0,1,2,3,4,5,6,7,8,9}
        \draw (\x cm,1pt) -- (\x cm,-1pt) node[anchor=north] {$\x$};

        \filldraw[fill=red!40!white, draw=gray] (2,0) rectangle (3,1);
        \filldraw[fill=red!40!white, draw=gray] (3,0) rectangle (5,1.7);
        \filldraw[fill=red!40!white, draw=gray] (5,0) rectangle (6,2);
        \filldraw[fill=red!40!white, draw=gray] (6,0) rectangle (7.5,2.3);
        
        \shade[top color=blue,bottom color=gray!50] 
      (5,2.5) parabola (6,2.4) |- (5,2);
      
        \shade[top color=blue,bottom color=gray!50] 
              (6,3.5) parabola (7.5,2.8) |- (6,2.3);
              
        \shade[top color=blue,bottom color=gray!50] 
              (7.5,0.5) parabola (8,0.3) |- (7.5,0);
              
        \shade[top color=blue,bottom color=gray!50] 
              (8,1.5) parabola (9,0) |- (8,0);
        \end{tikzpicture}
        \caption{A common deadline instance with $\eta>\lambda$. The first time point with $\eta_i > \lambda $ is time $5$ in which we start to run qOA for the rest of the jobs (blue part) while we continue running $YDS$ for the jobs released before $5$ (red part). At time point $7.5$, the workload of the first set of jobs is finished.
}
        \label{fig:fig2}
    \end{figure}
    
\begin{restatable}{observation}{CommonFeasibility}
\label{obs:common-feasibility}
        Algorithm \textsc{CDSwP} fully processes the whole processing volume of each job $w_j$, within $[r_j,d]$.
 \end{restatable}
 \begin{proof}
    Note that by the algorithm definition, no job starts being processed before its arrival in any mode. So it suffices to show that the complete processing volume of each job is completed before its deadline.
    Assume first that the algorithm remains in the FtP-mode until $d$.
    By the definition of the job instances $J^i$, any still unfinished processing volume $w_{n'}$ will be assigned to job $n'$ at timepoint $r_n$ and YDS will schedule it within $[r_n,d)$ according to YDS at a speed of $w_{n'}/(d-r_n)$. So the resulting schedule is feasible in that case. If the algorithm switches to the recovery mode at some $r_k$, then by the above argument the speed profile $s(FtP(k-1),t)$ is sufficient to finish  jobs $1,\dots k-1$, and furthermore speed profile $s(qOA(k),t)$ is feasible for for jobs $k,\dots n$, by the feasibility of algorithm $qOA$. So the overall speed profile $s(FtP(k-1),t)+s(qOA(j),t)$ is sufficient for processing the whole volume.
    \end{proof}

We begin by showing the following theorem which will imply consistency and smoothness. 
\begin{lemma}[Consistency \& Smoothness]\label{GoodPrediction}
Under the assumption that $\eta \in (0, \lambda)$, there holds
\begin{align*}
    \mathcal{E}_{CDSwP}\leq \left(\frac{1+\eta}{1-\lambda}\right)^{\alpha-1}  \cdot  \mathcal{E}_{YDS(J_{R})}.
\end{align*}
\end{lemma}

Before proving Lemma~\ref{GoodPrediction} we show the following intermediate result.

\begin{restatable}{lemma}{CompareToShrinkingPrediction}
\label{CompareToShrinkingPrediction}
Assuming that $\eta \in (0, \lambda)$, there holds
\begin{align*}
    \mathcal{E}_{CDSwP}\leq  \mathcal{E}_{YDS(J_{P'})}
\end{align*}
\end{restatable}
\begin{proof}
    Consider job instance ${J^{i}}'$ which consists of:
    \begin{itemize}
        \item A job $j^{i-1}$ with release time $r_i$, deadline $d$ and volume $w_{j^{i-1}} := w_i'-w_i$ equal to the total volume of jobs $1,\dots i-1$ that is still unfinished at $r_i$,
        \item Job $i$ with release time at $p_i'$ (and still deadline $d$ and processing volume $w_i$),
        \item For each job $j$ not yet released at $r_j$, include job $j$ with a release time of $p'_j$, a deadline of $d$ and a volume of $w_i$ in ${J^{i}}'$.
    \end{itemize}
    
    Note that, instance ${J^{i}}'$ differs from $J^i$ only in that job $i$ is considered separately, and not together with all previously released jobs that are still not finished. By Observation~\ref{obs:pjrjqjdj} a YDS schedule for the former is a feasible schedule for the later, and therefore by optimality of YDS,
    
    \begin{align}
    \label{eq:main22} 
        \mathcal{E}_{CDSwP(J^i)}^{[r_i, \infty)} \leq  \mathcal{E}_{CDSwP({J^{i}}')}^{[r_i, \infty)},
    \end{align}
    where $\mathcal{E}_{A(J)}^{[a,b]}$, refers to the energy consumption that the schedule produced by algorithm A on instance $J$ has within interval $[a,b]$.

    Using this notation, we can express the total energy-consumption of the $CDSwP$ as

    \begin{align*}
        \mathcal{E}_{CDSwP} &= \sum_{i=1}^n \mathcal{E}_{CDSwP(J^i)}^{[r_i, r_{i+1})}\\ 
        &= \sum_{i=1}^{n-2} \mathcal{E}_{CDSwP(J^i)}^{[r_i, r_{i+1})} +  \mathcal{E}_{CDSwP(J^{n-1})}^{[r_{n-1}, r_n)} +
    \mathcal{E}_{CDSwP(J^{n})}^{[r_{n}, d)}\\
        &\leq \sum_{i=1}^{n-2} \mathcal{E}_{CDSwP(J^i)}^{[r_i, r_{i+1})} +  \mathcal{E}_{CDSwP(J^{n-1})}^{[r_{n-1}, r_n)} +
        \mathcal{E}_{CDSwP({J^{n}}')}^{[r_{n}, d)}\\
        &=  \sum_{i=1}^{n-2} \mathcal{E}_{CDSwP(J^i)}^{[r_i, r_{i+1})} +  \mathcal{E}_{CDSwP(J^{n-1})}^{[r_{n-1},d)}\\
        &\vdots \\
         &\le \mathcal{E}_{CDSwP(J^1)}^{[r_1, d)} \leq \mathcal{E}_{YDS(J_{P'})},
    \end{align*}
where the inequalities follow by applying Equation~(\ref{eq:main22}).
\end{proof}

\begin{proof}[Proof of Lemma~\ref{GoodPrediction}]

By combining Lemmas~\ref{CompareToShrinkingPrediction} and~\ref{CD_shrunkPredictionError} we have,
\begin{align*}
    \mathcal{E}_{CDSwP}\leq   \mathcal{E}_{YDS(J_{P'})} \leq \left(\frac{1+ \eta}{1-\lambda}\right)^{\alpha-1}  \cdot  \mathcal{E}_{YDS(J_{R})},
\end{align*}

and the lemma directly follows. 
\end{proof}

We note that the above proof also works in exactly the same way when only a subset $A$ of the job set is processed.

\begin{corollary}
\label{cor:good}
Consider a set of jobs $A\subseteq \mathcal{J}$ and assume that $\eta_i \in (0, \lambda)$ holds for every job $i\in A$. Then
\begin{align*}
    \mathcal{E}_{CDSwP(A)}\leq  \mathcal{E}_{YDS(J_{P'}(A))}
    \end{align*}
\end{corollary}

We next analyze the case of inadequate predictions.

\begin{restatable}{lemma}{Robustness}
(Robustness) 
\label{BadPredictionCommon}
With a parameter $\eta \notin (0, \lambda)$, we have
\begin{align*}
    \mathcal{E}_{CDSwP}\leq 2 ^ {\alpha}  \left(\frac{1+\lambda}{1-\lambda}\right)^{\alpha-1} \cdot  \mathcal{E}_{qOA}.
\end{align*}
\end{restatable}
\begin{proof}
As in the definition of \textsc{CDSwP}, let $k$ be the smallest index, such that $\eta_k>\lambda$. Hence, the algorithm switches to the recovery mode at $r_k$.
We partition the job set into two subsets $A= \{1, \cdots, {k-1}\}$ and $B= \{k, \cdots, n\}$.
By Lemma~\ref{lem:merge}, and by the fact that by Corollary~\ref{cor:good} the energy consumption for set B is at most the energy consumption of $qOA$ for the whole job instance, it suffices to upper bound the energy consumption required for set A by the total energy that $qOA(k)$ uses.

We transform the schedule obtained by \textsc{CDSwP} for job set $A$ through three intermediate steps to the schedule produced by $YDS^J(R)$. Since $\mathcal{E}_{YDS^J(R)}\le \mathcal{E}_{qOA^J(R)}$ this will imply the theorem. 

\paragraph{Step 1:} 
Let $J^A$ be the job instance that contains all jobs in $A$, along with jobs $j=k,k+1,\dots n$ with respective release time $p_j'$, deadline $d$ and processing volume $w_j$.

Let $\mathcal{E}_{CDSwP}^A$, and $\mathcal{E}_{YDS(J^{A})}^A$ be the energy consumptions incured while scheduling the subset of jobs $A$ for  $\textsc{CDSwP(J)}$ and $YDS(J^{A})$ respectively.

By Corollary~\ref{cor:good}, 

\begin{align*}
    \mathcal{E}_{CDSwP}^A\leq \mathcal{E}_{YDS(J_{P'})}^A.
\end{align*}

Let $J^A_P$ be the job instance, consisting of the predicted release times ($p_i$) of jobs in set $A$ and the "shrunk" predicted release times ($p'_i$) for the remaining jobs. Note that $J^A_P$ differs from $J^A$ only in the release-times of jobs in set $A$. Since $\eta_i \le \lambda$ for any $i\in A$, there holds for any such $i$ that $d-p_i' = 1/(1-\lambda) (d-p_i)$. By Lemma~\ref{lem:ShrinkfromOneSide}, there therefore holds

\begin{align*}
    \mathcal{E}_{YDS(J_{P'})}^A\leq \left(\frac{1}{1-\lambda}\right)^{\alpha-1}  \cdot  \mathcal{E}_{YDS(J^{A}_P)}^A.
\end{align*}

Consider set $P^* = \{p^*_1,\dots p^*_{k-1}, p'_k, \dots, p'_n\}$ with $p^*_i = p_i - \eta_i (q_i - p_i)$ for all $j \in[k-1]$.
    
There holds
\begin{align}
    \label{eq:PredictReal}
    \mathcal{E}_{YDS(J^A_{P})} \le ( 1 + \lambda)^{\alpha -1} \mathcal{E}_{YDS(J_{P^*})} \le ( 1 + \lambda)^{\alpha -1} \mathcal{E}_{YDS(J^A)}.
\end{align}
    
    By having $c = \frac{1}{(1 + \lambda)}$, and $J= J_{P^*}$ $(r_i = p^*_i)$ in Lemma~\ref{lem:ShrinkfromOneSide}, we obtain $J' = J_{P}$. Since we have $\eta_j < \lambda$ for all $j <k$, the first inequality in~(\ref{eq:PredictReal}) holds.
    For every job $i\in A$ there holds $(r_i,d) \subseteq (p^*_i,d)$. 
    More specifically, a feasible schedule for $J^A$ is feasible for $J_{P^*}$ as well. The second inequality in~(\ref{eq:PredictReal}) then directly follows by the optimality of $YDS$.

Putting things together we therefore have

\begin{align}
    \label{eq:GoodpredictionPart}
    \mathcal{E}_{CDSwP}^A\leq \left(\frac{1+\lambda}{1-\lambda}\right)^{\alpha-1}  \cdot  \mathcal{E}_{YDS(J^{A})}^A,
\end{align}
\paragraph{Step 2:} 
In this step, we want to compare $\mathcal{E}_{YDS(J^{A})}^A$ with the energy of $YDS$ algorithm for a new job instance in which we consider the real release times for some jobs in set $B$ that their shrinking predictions are after their real release times.

A  job instance $J^l$ is defined, consisting of the real release times of jobs in set $A$, the real release times of job $j$ in set $B$ for which $r_j \leq p'_j$, and the shrunk prediction ($p'_j$) for the rest. Since moving the release times of the future jobs to the left could increase the speed (and hence increases energy) in the first part, 

\begin{align*}
   \mathcal{E}_{YDS(J^{A})}^A \leq \mathcal{E}_{YDS(J^{l})}^A.
\end{align*}

\paragraph{Step 3:}
In the last step, we want to compare $\mathcal{E}_{YDS(J^{l})}^A$ with the optimum offline algorithm ($YDS$) for the complete job instance $J$ and their real release time $J_R$. We want to show

\begin{align*}
    \mathcal{E}_{YDS(J^{l})}^A \leq \mathcal{E}_{YDS(J^{l})}^J \leq \mathcal{E}_{YDS(J_R)}^J.
\end{align*}

The first inequality holds because $A \subseteq J$. 
Consider the difference between two job instances $J_R$ and $J^l$.  Since for each job $i$, its available time in $J_R$ is a subset of its available time in $J^l$, $YDS(J_R)$ is a feasible algorithm for job instance $J^l$. Therefore, the second inequality holds.

So far we proved that

\begin{align*}
    \mathcal{E}_{CDSwP}^A\leq \left(\frac{1+\lambda}{1-\lambda}\right)^{\alpha-1}  \cdot  \mathcal{E}_{YDS(J_R)}^J.
\end{align*}

Since we run $qOA$ for the job set $B$, 

\begin{align*}
    \mathcal{E}_{CDSwP}^B =  \mathcal{E}_{qOA(J_R)}^B \leq \mathcal{E}_{qOA(J_R)}^J.
\end{align*}

And by Lemma \ref{lem:merge},  
    \begin{align*}
        \mathcal{E}_{CDSwP}\leq 2 ^ {\alpha-1} \cdot (\left(\frac{1+\lambda}{1-\lambda}\right)^{\alpha-1} \cdot \mathcal{E}_{YDS(J_R)}^J + \mathcal{E}_{qOA(J_R)}^J).
    \end{align*}
    
Since  $\mathcal{E}_{YDS(J_R)}^J \leq \mathcal{E}_{qOA(J_R)}^J$, 

     \begin{align*}
     \begin{split}
         \mathcal{E}_{CDSwP} \leq 
         2 ^ {\alpha-1} \cdot ( \mathcal{E}_{qOA}) (\left(\frac{1+\lambda}{1-\lambda}\right)^{\alpha-1} +1) \\ 
         \leq 
         2 ^ {\alpha}  \left(\frac{1+\lambda}{1-\lambda}\right)^{\alpha-1} \cdot ( \mathcal{E}_{qOA}).
    \end{split}
    \end{align*}
    
\end{proof}

Lemmas~\ref{GoodPrediction} and~\ref{BadPredictionCommon} together imply Theorem~\ref{thm:common-main}.

\section{Discussion on Confidence Parameters $\lambda$ and $\mu$}
\label{sec:experiments}
In order to give some intuition on how the confidence parameters $\mu$ and $\lambda$ affect the obtained performance guarantees of \textsc{SwP}, we perform some numerical experiments for different settings. Moreover, we compare our algorithm with the currently best-known online algorithm qOA and the optimum offline algorithm YDS using real-world data. All experiments were run on a typical laptop computer.

We only consider $\alpha =3$ for the experiments, as this is the typical value of $\alpha$ for real-world processors, see for example~\cite{CritchlowDS00,WiermanAT12}. Furthermore for qOA, we only consider $q= 2 - \frac{1}{\alpha}\approx 1.667$ since this is the value that minimizes the competitive ratio~\cite{BansalCKP12}.

The input data for our experiments is the same as in~\cite{DBLP:journals/corr/AbousamraBP13}. There, jobs are generated from http requests received on EPAs web-server. For practical reasons, we limit our input instances to the first 1000 jobs of their sample. 
In order to generate predictions for the input, we use a normal distribution with a mean of $0$, and a standard deviation of ${0.01, 0.05}$, or $0.1$. For each job, two samples from this distribution are taken and each of them is scaled by the real interval length of the job. The result is then added to each job's actual release time and deadline to obtain predictions for them. 

In order to illustrate the effect of parameters $\lambda$ and $\mu$, we run \textsc{SwP} for different combinations of these values. In particular we consider $\lambda = {0, 0.1, 0.2, 0.3}$ and 
$\mu = 0.1, 0.2, \cdots, 1$. 
Our results with standard deviations $0.01$, $0.05$, $0.1$ can be found in Figure~\ref{fig:experiment-0.01}, Figure~\ref{fig:experiment-main}, and Figure~\ref{fig:experiment-0.1} respectively.

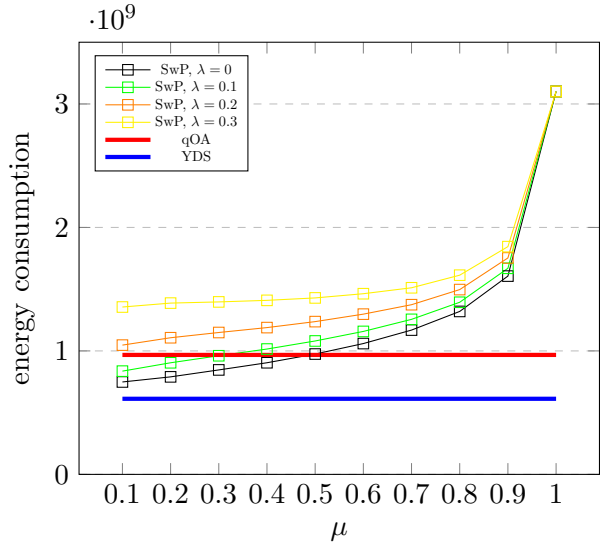
\begin{figure}
    \centering
\begin{tikzpicture}
\pgfplotsset{width=8.5cm,compat=1.9}
\begin{axis}[
    xlabel={$\mu$},
    ylabel={energy consumption},
    ymin=0, ymax=3500000000,
    xtick={0, 0.1, 0.2, 0.3, 0.4, 0.5, 0.6, 0.7, 0.8, 0.9, 1},
    legend pos=north west,
    ymajorgrids=true,
    grid style=dashed,
    legend style={nodes={scale=0.5, transform shape}}, 
]

\addplot[
    color=black,
    mark=square,
    ]
    coordinates {
        (0.1, 748456610.061912)(0.2, 789821792.877815)(0.30000000000000004, 846029193.083950)(0.4, 904149412.586250)(0.5, 974653003.102324)(0.6000000000000001, 1059555054.642507)(0.7000000000000001, 1168756012.879264)(0.8, 1319746942.895084)(0.9, 1606029739.372549)(1.0, 3101125903.818738)
    };
    \addlegendentry{SwP, $\lambda = 0$}

\addplot[
    color=green,
    mark=square,
    ]
    coordinates {
        (0.1, 836258854.890122)(0.2, 903891436.025233)(0.30000000000000004, 960839752.770661)(0.4, 1015187056.160013)(0.5, 1080502562.130370)(0.6000000000000001, 1157836000.557006)(0.7000000000000001, 1255622023.960963)(0.8, 1394009311.352384)(0.9, 1669108392.740638)(1.0, 3101128780.882705)
    };
    \addlegendentry{SwP, $\lambda = 0.1$}
    
\addplot[
    color=orange,
    mark=square,
    ]
    coordinates {
        (0.1, 1047011729.028555)(0.2, 1106700412.498293)(0.30000000000000004, 1149082727.886969)(0.4, 1188873790.459933)(0.5, 1237261068.054246)(0.6000000000000001, 1297834721.482126)(0.7000000000000001, 1374221757.143830)(0.8, 1496677478.356415)(0.9, 1752788422.771746)(1.0, 3101140689.952669)
    };
    \addlegendentry{SwP, $\lambda = 0.2$}
    
\addplot[
    color=yellow,
    mark=square,
    ]
    coordinates {
        (0.1, 1355470178.454708)(0.2, 1387282708.304186)(0.30000000000000004, 1396834898.223269)(0.4, 1409522879.424575)(0.5, 1429107915.423332)(0.6000000000000001, 1463256239.607983)(0.7000000000000001, 1511349570.978009)(0.8, 1613391898.748042)(0.9, 1844462620.203460)(1.0, 3101125184.987737)
    };
    \addlegendentry{SwP, $\lambda = 0.3$}
    
\addplot[
    color=red,
    no marks,
    ultra thick,
    ]
    coordinates {
    (0.1, 967518127.858653)(0.2, 967518127.858653)(0.3, 967518127.858653)(0.4, 967518127.858653)(0.5, 967518127.858653)(0.6, 967518127.858653)(0.7, 967518127.858653)(0.8, 967518127.858653)(0.9, 967518127.858653)(1.0, 967518127.858653)
    };
    \addlegendentry{qOA}
    
\addplot[
    color=blue,
    no marks,
    ultra thick,
    ]
    coordinates {
    (0.1, 612314285.007689)(0.2, 612314285.007689)(0.3, 612314285.007689)(0.4, 612314285.007689)(0.5, 612314285.007689)(0.6, 612314285.007689)(0.7, 612314285.007689)(0.8, 612314285.007689)(0.9, 612314285.007689)(1.0, 612314285.007689)
    };
    \addlegendentry{YDS}
    
\end{axis}
\end{tikzpicture}
    \caption{Prediction set 1, stddev=0.01}
    \label{fig:experiment-0.01}
\end{figure}

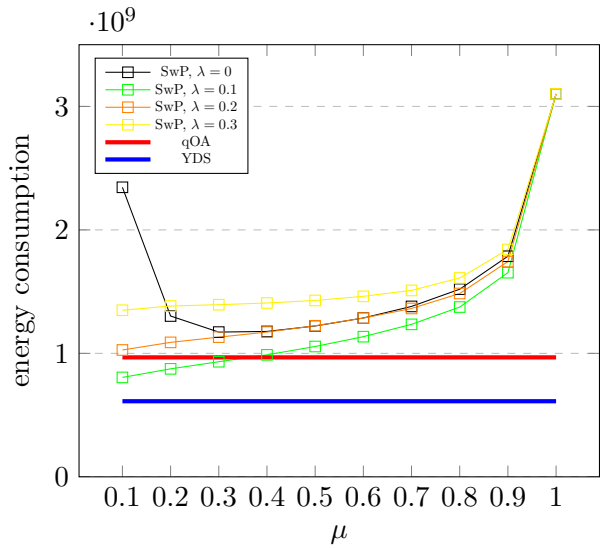
\begin{figure}
    \centering
   \begin{tikzpicture}
\pgfplotsset{width=8.5cm,compat=1.9}
\begin{axis}[
    xlabel={$\mu$},
    ylabel={energy consumption},
    ymin=0, ymax=3500000000,
    xtick={0, 0.1, 0.2, 0.3, 0.4, 0.5, 0.6, 0.7, 0.8, 0.9, 1},
    legend pos=north west,
    ymajorgrids=true,
    grid style=dashed,
    legend style={nodes={scale=0.5, transform shape}},
]

\addplot[
    color=black,
    mark=square,
    ]
    coordinates {
        (0.1, 2345047856.374099)(0.2, 1302424319.153903)(0.30000000000000004, 1172341591.728934)(0.4, 1177627676.361933)(0.5, 1221765014.784079)(0.6000000000000001, 1286503484.810952)(0.7000000000000001, 1380200372.010677)(0.8, 1519624397.298565)(0.9, 1786558196.564168)(1.0, 3101125394.868737)
    };
    \addlegendentry{SwP, $\lambda = 0$}

\addplot[
    color=green,
    mark=square,
    ]
    coordinates {
        (0.1, 805317522.085589)(0.2, 874654869.263613)(0.30000000000000004, 932348502.615183)(0.4, 988034797.892609)(0.5, 1055648590.575977)(0.6000000000000001, 1135146811.699021)(0.7000000000000001, 1234963776.025168)(0.8, 1375206598.353594)(0.9, 1653488046.298736)(1.0, 3101131708.213713)
    };
    \addlegendentry{SwP, $\lambda = 0.1$}
    
\addplot[
    color=orange,
    mark=square,
    ]
    coordinates {
        (0.1, 1027449262.752904)(0.2, 1089058964.744987)(0.30000000000000004, 1131490732.596269)(0.4, 1172364152.656896)(0.5, 1223742333.209896)(0.6000000000000001, 1286054494.228102)(0.7000000000000001, 1363517716.252757)(0.8, 1484866633.151484)(0.9, 1743627246.571100)(1.0, 3101142940.139664)
    };
    \addlegendentry{SwP, $\lambda = 0.2$}
    
\addplot[
    color=yellow,
    mark=square,
    ]
    coordinates {
        (0.1, 1349163337.998995)(0.2, 1383624722.467036)(0.30000000000000004, 1394132377.344509)(0.4, 1408122111.487158)(0.5, 1428498563.162338)(0.6000000000000001, 1461913073.655294)(0.7000000000000001, 1509877499.393803)(0.8, 1610545649.808234)(0.9, 1841603899.205809)(1.0, 3101127699.153709)
    };
    \addlegendentry{SwP, $\lambda = 0.3$}
    
\addplot[
    color=red,
    no marks,
    ultra thick,
    ]
    coordinates {
    (0.1, 967518127.858653)(0.2, 967518127.858653)(0.3, 967518127.858653)(0.4, 967518127.858653)(0.5, 967518127.858653)(0.6, 967518127.858653)(0.7, 967518127.858653)(0.8, 967518127.858653)(0.9, 967518127.858653)(1.0, 967518127.858653)
    };
    \addlegendentry{qOA}
    
\addplot[
    color=blue,
    no marks,
    ultra thick,
    ]
    coordinates {
    (0.1, 612314285.007689)(0.2, 612314285.007689)(0.3, 612314285.007689)(0.4, 612314285.007689)(0.5, 612314285.007689)(0.6, 612314285.007689)(0.7, 612314285.007689)(0.8, 612314285.007689)(0.9, 612314285.007689)(1.0, 612314285.007689)
    };
    \addlegendentry{YDS}
    
\end{axis}
\end{tikzpicture}
    \caption{Prediction set 2, stddev=$0.05$.}
     \label{fig:experiment-main}
\end{figure}

\begin{figure}
    \centering
\begin{tikzpicture}
\pgfplotsset{width=8.5cm,compat=1.9}
\begin{axis}[
    xlabel={$\mu$},
    ylabel={energy consumption},
    ymin=0, ymax=3500000000,
    xtick={0, 0.1, 0.2, 0.3, 0.4, 0.5, 0.6, 0.7, 0.8, 0.9, 1},
    legend pos=north west,
    ymajorgrids=true,
    grid style=dashed,
    legend style={nodes={scale=0.5, transform shape}},
]

\addplot[
    color=black,
    mark=square,
    ]
    coordinates {
        (0.1, 11119169852.588633)(0.2, 3629321670.747347)(0.30000000000000004, 2261840172.913717)(0.4, 1854525358.878197)(0.5, 1729371524.119441)(0.6000000000000001, 1717177348.892627)(0.7000000000000001, 1751324635.929418)(0.8, 1842245801.109186)(0.9, 2050140413.864959)(1.0, 3101125184.987737)
    };
    \addlegendentry{SwP, $\lambda = 0$}

\addplot[
    color=green,
    mark=square,
    ]
    coordinates {
        (0.1, 1059337962.833822)(0.2, 976548487.583714)(0.30000000000000004, 1016676750.464443)(0.4, 1065393011.979906)(0.5, 1122478925.305598)(0.6000000000000001, 1194772783.318888)(0.7000000000000001, 1290947931.190356)(0.8, 1434170751.525954)(0.9, 1703677598.912743)(1.0, 3101136078.389686)
    };
    \addlegendentry{SwP, $\lambda = 0.1$}
    
\addplot[
    color=orange,
    mark=square,
    ]
    coordinates {
        (0.1, 958659612.241536)(0.2, 1027928593.055258)(0.30000000000000004, 1082410598.291959)(0.4, 1131365106.115680)(0.5, 1184008106.089633)(0.6000000000000001, 1251749102.497711)(0.7000000000000001, 1339260849.393518)(0.8, 1473558285.954850)(0.9, 1737944712.729144)(1.0, 3101125903.818738)
    };
    \addlegendentry{SwP, $\lambda = 0.2$}
    
\addplot[
    color=yellow,
    mark=square,
    ]
    coordinates {
        (0.1, 1273578479.620560)(0.2, 1322178707.810232)(0.30000000000000004, 1342100448.506997)(0.4, 1364924742.185967)(0.5, 1391524858.245256)(0.6000000000000001, 1432116551.018006)(0.7000000000000001, 1491966331.702830)(0.8, 1599678008.800989)(0.9, 1837872497.424757)(1.0, 3101137878.949682)
    };
    \addlegendentry{SwP, $\lambda = 0.3$}
    
\addplot[
    color=red,
    no marks,
    ultra thick,
    ]
    coordinates {
    (0.1, 967518127.858653)(0.2, 967518127.858653)(0.3, 967518127.858653)(0.4, 967518127.858653)(0.5, 967518127.858653)(0.6, 967518127.858653)(0.7, 967518127.858653)(0.8, 967518127.858653)(0.9, 967518127.858653)(1.0, 967518127.858653)
    };
    \addlegendentry{qOA}
    
\addplot[
    color=blue,
    no marks,
    ultra thick,
    ]
    coordinates {
    (0.1, 612314285.007689)(0.2, 612314285.007689)(0.3, 612314285.007689)(0.4, 612314285.007689)(0.5, 612314285.007689)(0.6, 612314285.007689)(0.7, 612314285.007689)(0.8, 612314285.007689)(0.9, 612314285.007689)(1.0, 612314285.007689)
    };
    \addlegendentry{YDS}
    
\end{axis}
\end{tikzpicture}
    \caption{Prediction set 3, stddev=0.1}
    \label{fig:experiment-0.1}
\end{figure}
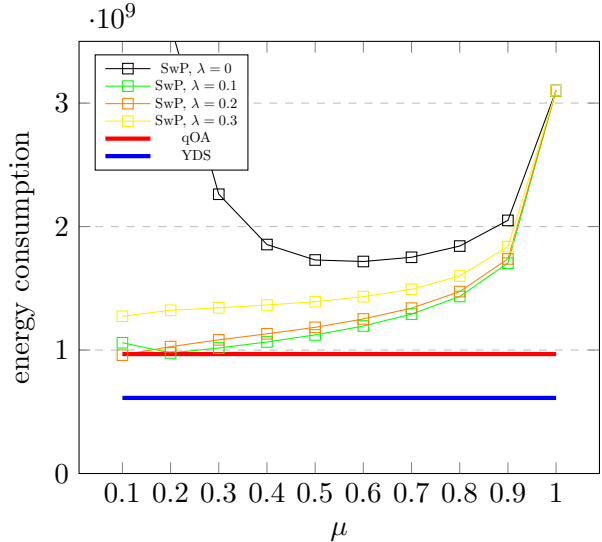

To gain some intuition on the results, recall that $\mu$ denotes the portion of each block for which AVR is run. In particular, for $\mu=1$ the SwP algorithm becomes identical to the AVR algorithm and disregards the predictions, whereas the smaller $\mu$'s value the more the predictions are trusted. This explains why the competitive ratio increases with $\mu$. Similarly, recall that $\lambda$ defines how much the predicted interval will be shrunk and that the improved competitive ratio is only proven for $\eta \le \lambda$ but on the other hand the bigger $\lambda$ gets the smaller that improvement in the competitive ratio will be. Although the best choices for $\lambda$ and $\mu$ depend on the quality and/or structure of the predictions, our experiments highlight that for appropriate such choices, one can significantly improve upon the energy-consumption of qOA. To summarize, in practice the most sensible settings of $\lambda$ and $\mu$ will depend on the quality as well as structure of the predictions and it may be worthwhile experimenting with different such settings.

\section{Conclusion}
In this paper, we have presented a consistent, smooth, and robust algorithm for the general classical, deadline-based, online speed-scaling problem using ML predictions for release times and deadlines.

We can remove the assumption of knowing the number of jobs $n$, by slightly adapting the error definition, so that the prediction is considered to be inadequate if the predicted number of jobs is wrong.

It remains an interesting open question on whether a similar robust, consistent and smooth algorithm exists for the more general setup in which the workloads of the jobs are not known in advance but predicted along with their release times and deadlines. Although we were able to extend \textsc{SwP} under the assumption that it satisfies a natural monotonicity property, it is unclear if that property holds in general.

\bibliography{references}
\bibliographystyle{abbrv}
\pagebreak

\appendix
\section{Energy of the Shrunk Instances}

\ShrinkfromOneSide*
\begin{proof}
    In order to prove the theorem, we start with schedule $C_{YDS}(\mathcal{J})$, in which each job $j_i$ runs at a speed $s_i$. It suffices to show that there exists a feasible to schedule $C'(\mathcal{\hat{J}})$ for the jobs of $\hat{J}$ in which each job $\hat{j}_i$ runs at speed of $\hat{s}_i = \frac{1}{c}\cdot s_i$. The theorem then directly follows by the definitions of the energy and the power function.
    
    It is without loss of generality to assume that both $C_{YDS}(\mathcal{J})$ and $C'(\mathcal{\hat{J}})$ are earliest release time first schedules.
    Let $a_i$ be the timepoint at which $j_i$ starts being processed in $C_{YDS}(\mathcal{J})$, and let $\hat{a}_i = d - c(d-a_i) = d(1-c) + ca_i$. Note that by construction, feasibility of $C_{YDS}(\mathcal{J})$ and because $c<1$, we have that $a_i \ge r_i$ and therefore $\hat{a}_i = d(1-c) + ca_i \ge d(1-c) + cr_i = \hat{r}_i$. Consider earliest release time first schedule $C(\mathcal{\hat{J}})$ in which every job $j_i$ is processed at a speed of  $\hat{s}_i$. 
    
    In the remainder of this proof we show by induction, that in $C'(\mathcal{\hat{J}})$ every job $\hat{j}_i$ starts at $\hat{a}_i$. For the base case, it can be easily shown than in the optimal EDF schedule $j_1$ starts at $r_1$ and therefore $a_1=r_1$. It follows that $\hat{a}_1=\hat{r}_1$ and the job can feasibly start executing at $\hat{a}_1$. Assume that in $C'(\mathcal{\hat{J}})$ the first (according to release time) $i$ many jobs start at their respective $\hat{a}_i$'s. Then, in order to show that $\hat{j}_{i+1}$ starts at $\hat{a}_{i+1}$ it suffices to show that $\hat{j}_i$ finishes its execution before $\hat{a}_{i+1}$. Clearly $j_i$ finishes its execution no later than $a_{i+1}$. Therefore, $a_i+w_i/s_i \le a_{i+1}$. In turn,
        \begin{align*}
            \hat{a}_i + w_i/s_i' = 
            d_i(1-c) + ca_i + c(w_i/s_i) \le \\
            d_i(1-c) + ca_{i+1} = \hat{a}_{i+1}.
        \end{align*}
        and $\hat{j}_{i+1}$ can start at $\hat{a}_{i+1}$. A similar argument shows that $\hat{j}_n$ finishes before $\hat{d}$ and feasibility and therefore the proof of the theorem follows.
\end{proof}

\ShrinkingRatio*
\begin{proof}
    We may assume that $C_{YDS}(\mathcal{J})$ is an earliest deadline first schedule in which every job $j_i$ runs at a speed of $s_i$. It suffices to show that there exists a feasible schedule $C'(\mathcal{\hat{J}})$ for the jobs of $\hat{J}$ in which each job $\hat{j}_i$ runs at a speed of $\hat{s}_i = \frac{1}{c}\cdot s_i$. Let the jobs be ordered by their deadlines.
    
    Consider an earliest deadline first schedule $C'(\mathcal{\hat{J}})$ in which every job $\hat{j}_i$ runs at a speed of $\hat{s}_i$. For a job $k$ and a job $i\le k$, let $t_i(k)$ be the amount of time during which $i$ is processed within $(r_k,d_k)$. Similarly we can define $\hat{t}_i(k)$. It suffices to show that for every job $j_k$ and every $i<k$ there holds 
    \begin{align*}
         \hat{t}_i(k) \le c \cdot t_i(k).
    \end{align*}
    Assume for the sake of contradiction that this does not hold, and let $j_k$ and $j_i$ be the first pair of jobs for which it is not true. For this to be the case we must have $r_i\le r_k\le d_i\le d_k$ and $\hat{r}_i\le \hat{r}_k\le \hat{d}_i\le \hat{d}_k$. The reason is that if the intervals are disjoint or the interval of one job contains the other in any schedule then the inequality directly holds. So let $z = r_k-r_i$ and $\hat{z}=\hat{r}_k-\hat{r}_i$. By construction it must be that $\hat{z}\ge c z$. Therefore and by the assumption that $k$ and $i$ are the first such pair, more of job $\hat{j}_i$ runs in $\hat{z}$ than of job $j_i$ in $z$. This leads to a contradiction since there is less of job $\hat{j}_i$ left to run after $\hat{r}_k$
\end{proof}

\section{Calculating the $y_i^t$'s}
\label{app:y}

First we show the following lemma.

\begin{lemma}
\label{lem:y_technical} 
    For any given $0 \le X \le \ell(j) \max_{t\in[r_j,d_j)} (V_t(j)+\delta_j)/\mu$ there exist values $y_j^t$, with $0\le y_j^t\le \delta_j$ so that equations~(\ref{eq:y_between}),(\ref{eq:y_zero}) and~(\ref{eq:y_full}) are satisfied for all $t\in[r_j,d_j)$, with $X$ in place of $X_j$. Furthermore for any $t,t'$ with $y_j^t \le y_j^{t'}$ there holds $V_{t'}(j)\le V_t(j)$, and $\sum y_j^t$ is a continuous and non-decreasing function in $X$.
\end{lemma}
\begin{proof}
    If $X/\ell(j)<\min_{t\in[r_j,d_j)} V_t(j)/\mu$, then it is easy to verify that $y_j^t=\delta_j$ for all $t\in[r_j,d_j)$ satisfies all equations. So we assume for the remainder of this proof that $\min_{t\in [r_j,d_j)}V_t(j)/\mu \le X/\ell(j) \le \max_{t\in[r_j,d_j)} (V_t(j)+\delta_j)/\mu$.
    
    For any $t\in[r_j,d_j)$, let 
    \begin{equation}
    \label{eq:ys}
    y_j^t := 
    \begin{cases}
         0, &\text{ if } V_t(j)/\mu\ge X/\ell(j),\\
         \delta_j, &\text{ if } (V_t(j) +\delta_j)/\mu \le X/\ell(j),\\
          \mu X/\ell(j) - V_t(j), &\text{ otherwise}.
    \end{cases}
    \end{equation}
    It is easy to verify that for the above definition of $y_j^t$, equations (\ref{eq:y_between}), (\ref{eq:y_zero}) and (\ref{eq:y_full}) are satisfied with $X$ in place of $X_j$, and that for any $t,t'$ with $y_j^t \le y_j^{t'}$ there holds $V_{t'}(j)\le V_t(j)$. Finally, $\sum y_j^t$ is a continuous function as a sum of a finite number of continuous functions, and non-decreasing in $X$ (as each $y_j^t$ is by definition a non-increasing function of $X$).
\end{proof}

\begin{lemma}
\label{lem:y_exist}
For any set of values $V_t(j)$, there exist values $y_j^t$, with $0\le y_j^t\le \delta_j$ so that equations~(\ref{eq:y_between}),(\ref{eq:y_zero}) and~(\ref{eq:y_full}) are satisfied for all $t\in[r_j,d_j)$.
\end{lemma}
\begin{proof}
Note that it suffices to show that there exists $X_j = w_j - \sum_t y_j^t$ where the $y_j^t$ are as defined in the proof of Lemma~\ref{lem:y_technical}, since then by Lemma~\ref{lem:y_technical} the equations (\ref{eq:y_between}),(\ref{eq:y_zero}), and (\ref{eq:y_full}) would hold for $X_j = w_j - \sum_t y_j^t$.

First, let $X=w_j$ and compute the values of $y_j^t$ via (\ref{eq:ys}). If $\sum_t y_j^t = 0$ , then we have found the desired $X$ and are done. Assume therefore, that $0 < \sum_t y_j^t \le w_j$. By Lemma~\ref{lem:y_technical}, $\sum_t y_j^t$ is a non-decreasing and continuous function of $X$ within $[0,w_j]$ that obtains value $0$ for $X=0$, and a value $\le w_j$ for $X=w_j$. Equivalently the function $w_j-\sum_t y_j^t$ is non-increasing and continuous in $X$ within $[0,w_j]$ and obtains value $w_j$ for $X=0$ and a value $\ge 0$ for $X=w_j$. Therefore, by the intermediate value theorem there must exist an $X_j\in[0,w_j]$, such that $w_j - \sum_t y_j^t$ obtains a value of $X_j$, which concludes the proof of the lemma.
\end{proof}

\subsection{Algorithm}
Lemmas~\ref{lem:y_technical} and~\ref{lem:y_exist} directly imply an algorithm for identifying such values of $y_j^t$. In particular, since for any $t,t'$ with $y_j^t \le y_j^{t'}$ there holds $V_{t'}(j)\le V_t(j)$, we can order all relevant $t$'s by $V_t(j)$ and find (through enumeration) $t',t''$ such that for any $V_t(j)\ge V_t'(j)$ we have $y_j^t=0$, for any $V_t(j)\le V_{t''}(j)$, $y_j^t =\delta_j$ and for all other $t$ there holds $0< y_j^t<\delta_j$. Let $N$ be the number of $t$'s such that $y_j^t=\delta_j$, and $Z= w_j-N\delta_j$ be the remaining processing volume that needs to be assigned through the $y_j^t$'s for $t$'s with $V_{t''}(j) < V_t(j) < V_{t'}(j)$. In other words we need to find $0< y_j^t<\delta_j$ so that $Z - \sum_t y_j^t = X_j$, and for each individual such $y_j^t$, we have $y_j^t = \mu X_j/\ell(j) - V_t(j)$. This implies a system of $k+1$ equations (for some $k$) with $k+1$ unknowns, that by Lemma~\ref{lem:y_exist} has a solution assuming that $t',t''$ were chosen correctly.

\section{Combining Online Scheduling Algorithms}

The proof of the following lemma is an adaptation of the proof of a similar result used in the analysis of the Average Rate algorithm for the problem (see~\cite{YDS,Bansal-AVR}). However we reprove it here, since (i) we obtain a more general result, and (ii) our respective schedules do not necessarily satisfy all the properties of the corresponding schedules in \cite{YDS,Bansal-AVR}.

\merge*
\begin{proof}
    Regarding feasibility, consider a partitioning of the time horizon defined by the points $T = \cup_i\{r_i,d_i\}$. Let $t_1,t_2,\dots$ be the points of $T$ ordered from left to right. By definition of $s_{C}(t)$ we have that $\int_{t_i}^{t_{i+1}}s_{C}(t)dt =  \sum_i\int_{t_i}^{t_{i+1}} s_i(t)$. Consider schedule $C$ that  processes in every interval $(t_i,t_{i+1})$ the same amount of volume for each job $j$ as the corresponding schedule $C_k$ with $j\in C_k$ does in this interval. The jobs inside such an interval $(t_i,t_{i+1})$ are processed in an arbitrary order (this is possible because the interval does not contain any release times or deadlines). Assume for the sake of contradiction that $C$ is not feasible for $J$. Then there must exist a job $j\in C_k$ that misses its deadline $d_j = t_\ell$. This contradicts the feasibility of $C_k$ since $C$ processes exactly the same amount of job $j$ in every interval $(t_i,t_{i+1})$ for $i=0,\dots \ell-1$ as $C_k$.

    With respect to the energy consumption, we have:
    \begin{align*}
        \mathcal{E}_{C} = \int_t \left(\sum_i s_i(t)\right)^\alpha dt.
    \end{align*}
   
    Note that for all $i$, $\int_t s_{i}(t)^\alpha dt \le \mathcal{E}_{i}$.
    The lemma now follows since
    \begin{multline*}
        \int_t\left(\sum_i s_i(t)\right)^\alpha dt \le m^{\alpha-1}\int_t\left(\sum_i s_{i}(t)^\alpha\right) dt 
        =
        m^{\alpha-1}\left(\sum_i \mathcal{E}_{i}\right),
    \end{multline*}
    where the first inequality follows by Jensen's inequality.
\end{proof}

\end{document}